\documentclass[11pt,a4paper]{article}
\usepackage{hyperref}
\usepackage{authblk}

\title{On Comparable Box Dimension}

\author[1]{Zdenek Dvor\'ak\thanks{Supported by the ERC-CZ project LL2005 (Algorithms and complexity within and beyond bounded expansion) of the Ministry of Education of Czech Republic. \texttt{rakdver@iuuk.mff.cuni.cz}}}
\author[2]{Daniel Goncalves\thanks{Supported by the ANR grant GATO ANR-16-CE40-0009. \texttt{goncalves@lirmm.fr}}}
\author[1]{Abhiruk Lahiri\thanks{Supported by the ERC-CZ project LL2005 (Algorithms and complexity within and beyond bounded expansion) of the Ministry of Education of Czech Republic. \texttt{abhiruk@iuuk.mff.cuni.cz}}}
\author[3]{Jane Tan\thanks{\texttt{jane.tan@maths.ox.ac.uk}}}
\author[4]{Torsten Ueckerdt\thanks{\texttt{torsten.ueckerdt@kit.edu}}}
\affil[1]{Charles University, Prague, Czech Republic}
\affil[2]{LIRMM, Université de Montpellier, CNRS, Montpellier, France}
\affil[3]{Mathematical Institute, University of Oxford, Oxford OX2 6GG, United Kingdom}
\affil[4]{Karlsruhe Institute of Technology, Karlsruhe, Germany}

\usepackage{amsthm}
\usepackage{amsfonts}
\usepackage{amsmath}
\usepackage{amssymb}
\usepackage{colonequals}
\usepackage{epsfig}
\usepackage{url}
\newcommand{\GG}{{\cal G}}
\newcommand{\HH}{{\cal H}}

\newcommand{\OO}{{\cal O}}

\newcommand{\brm}[1]{\operatorname{#1}}
\newcommand{\cbdim}{\brm{dim}_{cb}}
\newcommand{\ecbdim}{\brm{dim}^{ext}_{cb}}
\newcommand{\tw}{\brm{tw}}
\newcommand{\vol}{\brm{vol}}

\newtheorem{theorem}{Theorem}
\newtheorem{lemma}{Lemma}
\newtheorem{corollary}{Corollary}
\newtheorem{definition}{Definition}
\begin{document}
\maketitle

\begin{abstract}
Two boxes in $\mathbb{R}^d$ are \emph{comparable} if one of them is a subset
of a translation of the other one. The \emph{comparable box dimension} of a graph
$G$ is the minimum integer $d$ such that $G$ can be represented as a
touching graph of comparable axis-aligned boxes in $\mathbb{R}^d$. We
show that proper minor-closed classes have bounded comparable box
dimension and explore further properties of this notion.
\end{abstract}

\section{Introduction}

Given a system $\OO$ of subsets of $\mathbb{R}^d$, we say that a graph $G$ is a \emph{touching graph of objects from $\OO$}
if there exists a function $f:V(G)\to \OO$ (called a \emph{touching representation by objects from $\OO$})
such that the interiors of $f(u)$ and $f(v)$ are disjoint for all distinct $u,v\in V(G)$, and $f(u)\cap f(v)\neq\emptyset$ if and only if $uv\in E(G)$.
Famously, Koebe~\cite{koebe} proved that a graph is planar if and only if it is a touching graph of balls in $\mathbb{R}^2$.
This result has motivated numerous strengthenings and variations (see \cite{graphsandgeom, sachs94} for some classical examples); most relevantly for us, Felsner and Francis~\cite{felsner2011contact} showed that every planar graph is a touching graph of cubes in $\mathbb{R}^3$.

An attractive feature of touching representations is that it is possible to represent graph classes that are sparse
(e.g., planar graphs, or more generally, graph classes with bounded expansion~\cite{nesbook}).
This is in contrast to general intersection representations where the represented class always includes arbitrarily large cliques.
Of course, whether the class of touching graphs of objects from $\OO$ is sparse or not depends on the system $\OO$.
For example, all complete bipartite graphs $K_{n,m}$ are touching graphs of boxes in $\mathbb{R}^3$, where the vertices in
one part are represented by $m\times 1\times 1$ boxes and the vertices of the other part are represented by $1\times n\times 1$
boxes (throughout the paper, by a \emph{box} we mean an axis-aligned one, i.e., the Cartesian product of closed intervals of non-zero length).
Dvo\v{r}\'ak, McCarty and Norin~\cite{subconvex} noticed that this issue disappears if we forbid such a combination of
long and wide boxes, a condition which can be expressed as follows. For two boxes $B_1$ and $B_2$, we write $B_1 \sqsubseteq B_2$ if $B_2$ contains a translate of $B_1$.
We say that $B_1$ and $B_2$ are \emph{comparable} if $B_1\sqsubseteq B_2$ or $B_2\sqsubseteq B_1$.
A \emph{touching representation by comparable boxes} of a graph $G$ is a touching representation $f$ by boxes
such that for every $u,v\in V(G)$, the boxes $f(u)$ and $f(v)$ are comparable. 
Let the \emph{comparable box dimension} $\cbdim(G)$ of a graph $G$ be the smallest integer $d$ such that $G$ has a touching representation by comparable boxes in $\mathbb{R}^d$.
Let us remark that the comparable box dimension of every graph $G$ is at most $|V(G)|$, see Section~\ref{sec-vertad} for details.
Then for a class $\GG$ of graphs, let $\cbdim(\GG)\colonequals\sup\{\cbdim(G):G\in\GG\}$. Note that $\cbdim(\GG)=\infty$ if the
comparable box dimension of graphs in $\GG$ is not bounded.

Dvo\v{r}\'ak, McCarty and Norin~\cite{subconvex} proved some basic properties of this notion.  In particular,
they showed that if a class $\GG$ has finite comparable box dimension, then it has polynomial strong coloring
numbers, which implies that $\GG$ has strongly sublinear separators.  They also provided an example showing
that for many functions $h$, the class of graphs with strong coloring numbers bounded by $h$ has infinite
comparable box dimension\footnote{In their construction $h(r)$ has to be at least 3, and has to tend to $+\infty$.}. Dvo\v{r}\'ak et al.~\cite{wcolig}
proved that graphs of comparable box dimension $3$ have exponential weak coloring numbers, giving the
first natural graph class with polynomial strong coloring numbers and superpolynomial weak coloring numbers
(the previous example is obtained by subdividing edges of every graph suitably many times~\cite{covcol}).

We show that the comparable box dimension behaves well under the operations of addition of apex vertices,
clique-sums, and taking subgraphs.  Together with known results on product structure~\cite{DJM+}, this implies
the main result of this paper.

\begin{theorem}\label{thm-minor}
The comparable box dimension of every proper minor-closed class of graphs is finite.
\end{theorem}

Additionally, we show that classes of graphs with finite comparable box dimension are fractionally treewidth-fragile.
This gives arbitrarily precise approximation algorithms for all monotone maximization problems that are
expressible in terms of distances between the solution vertices and tractable on graphs of bounded treewidth~\cite{distapx}
or expressible in the first-order logic~\cite{logapx}.

\section{Parameters}

In this section we bound some basic graph parameters in terms of comparable box dimension.
The first result bounds the clique number $\omega(G)$ in terms of $\cbdim(G)$.
\begin{lemma}\label{lemma-cliq}
For any graph $G$, we have $\omega(G)\le 2^{\cbdim(G)}$.
\end{lemma}
\begin{proof}
We may assume that $G$ has bounded comparable box dimension
witnessed by a representation $f$. To represent any clique $A = \{a_1,\ldots,a_w\}$ in $G$, the
corresponding boxes $f(a_1),\ldots,f(a_w)$ have pairwise non-empty
intersections.  Since axis-aligned boxes have the Helly property, there
is a point $p \in \mathbb{R}^d$ contained in $f(a_1) \cap \cdots \cap
f(a_w)$.  As each box is full-dimensional, its interior intersects at
least one of the $2^d$ orthants at $p$. At the same time, it follows from the definition
of a touching representation that $f(a_1),\ldots,f(a_d)$ have pairwise disjoint
interiors, and hence $w \leq 2^d$.
\end{proof}
Note that a clique with $2^d$ vertices has a touching representation by comparable boxes in $\mathbb{R}^d$,
where each vertex is a hypercube defined as the Cartesian product of intervals of form $[-1,0]$ or $[0,1]$.
Together with Lemma~\ref{lemma-cliq}, it follows that $\cbdim(K_{2^d})=d$.

In the following we consider the chromatic number $\chi(G)$, and two
of its variants.  An \emph{acyclic coloring} (resp. \emph{star coloring}) of a graph $G$ is a proper
coloring such that any two color classes induce a forest (resp. star forest, i.e., a forest in which each component is a star).  The \emph{acyclic chromatic number} $\chi_a(G)$ (resp. \emph{star chromatic
  number} $\chi_s(G)$) of $G$ is the minimum number of colors in an acyclic (resp. star)
coloring of $G$.  We will need the fact that all the variants of the chromatic number
are at most exponential in the comparable box dimension; this follows
from~\cite{subconvex}, although we include an argument to make the
dependence clear.
\begin{lemma}\label{lemma-chrom}
For any graph $G$ we have $\chi(G)\le 3^{\cbdim(G)}$, $\chi_a(G)\le 5^{\cbdim(G)}$ and $\chi_s(G) \le 2\cdot 9^{\cbdim(G)}$.
\end{lemma}
\begin{proof}
We focus on the star chromatic number and note that the chromatic number and the acyclic chromatic number may be bounded similarly.
Suppose that $G$ has comparable box dimension $d$ witnessed by a representation $f$, and let $v_1, \ldots, v_n$
be the vertices of $G$ written so that $\vol(f(v_1)) \geq \ldots \geq \vol(f(v_n))$.
Equivalently, we have $f(v_i)\sqsubseteq f(v_j)$ whenever $i>j$. Now define a greedy coloring $c$ so that $c(v_i)$ is
the smallest color such that $c(v_i)\neq c(v_j)$ for any $j<i$ for which either $v_jv_i\in E(G)$ or there
exists $m>j$ such that $v_jv_m,v_mv_i\in E(G)$. Note that this gives a star coloring, since a path on four vertices always contains a 3-vertex subpath of the form $v_{i_1}v_{i_2}v_{i_3}$ such that $i_1<i_2,i_3$ and our coloring procedure gives distinct colors to vertices forming such a path.

It remains to bound the number of colors used. Suppose we are coloring $v_i$. We shall bound the number of vertices
$v_j$ such that $j<i$ and such that there exists $m>i$ for which $v_jv_m,v_mv_i\in E(G)$. Let $B$ be the box obtained by scaling up $f(v_i)$ by a factor of 5 while keeping the same center. Since $f(v_m)\sqsubseteq f(v_i)\sqsubseteq f(v_j)$, there exists a translation $B_j$ of $f(v_i)$
contained in $f(v_j)\cap B$ (see Figure~\ref{fig:lowercolcount}). Two boxes $B_{j}$ and $B_{j'}$ for $j\neq j'$ have disjoint interiors since their intersection is contained in the intersection of the touching boxes $f(v_{j})$ and $f(v_{j'})$, and their interiors are also disjoint from $f(v_i)\subset B$. Thus the number of such indices $j$ is at most $\vol(B)/\vol(f(v_i))-1=5^d-1$.

A similar argument shows that the number of indices $m$ such that $m<i$ and $v_mv_i\in E(G)$ is at most $3^d-1$.
Consequently, the number of indices $j<i$ for which there exists $m$ such that $j<m<i$ and $v_jv_m,v_mv_i\in E(G)$
is at most $(3^d-1)^2$. This means that when choosing the color of $v_i$ greedily, we only need to avoid colors of at most $(5^d-1) + (3^d-1) + (3^d-1)^2$ vertices, so $2\cdot 9^d$ colors suffice.
\end{proof}

\begin{figure}
\centering
\includegraphics[scale=1]{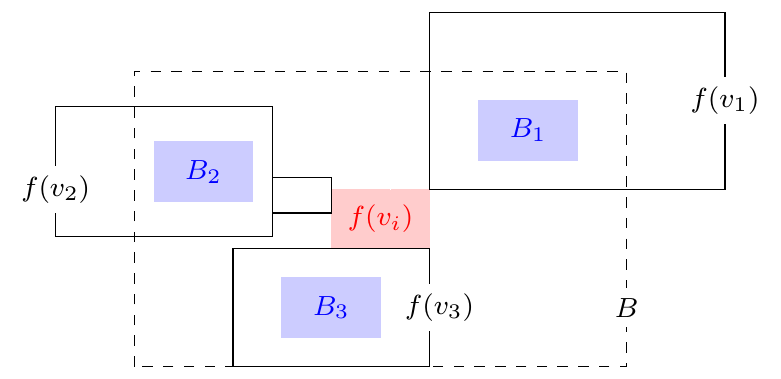}
\caption{Nearby boxes obstructing colors at $v_i$.}
\label{fig:lowercolcount}
\end{figure}

\section{Operations}

It is clear that given a touching representation of a graph $G$, one
can easily obtains a touching representation by boxes of an induced
subgraph $H$ of $G$ by simply deleting the boxes corresponding to the
vertices in $V(G)\setminus V(H)$.  In this section we are going to
consider other basic operations on graphs. In the following, to describe
the boxes, we are going to use the Cartesian product $\times$ defined among boxes ($A\times B$ is the box whose projection on the first dimensions gives the box $A$, while the projection on the remaining dimensions gives the box $B$) or we are going to provide its projections for every dimension ($A[i]$ is the interval obtained from projecting $A$ on its $i^\text{th}$ dimension).

\subsection{Vertex addition}\label{sec-vertad}

Let us start with a simple lemma saying that the addition of a vertex
increases the comparable box dimension by at most one.  In particular,
this implies that $\cbdim(G)\le |V(G)|$.
\begin{lemma}\label{lemma-apex}
For any graph $G$ and $v\in V(G)$, we have $\cbdim(G)\le \cbdim(G-v)+1$.
\end{lemma}
\begin{proof}
Let $f$ be a touching representation of $G-v$ by comparable boxes in $\mathbb{R}^d$, where $d=\cbdim(G-v)$.
We define a representation $h$ of $G$ as follows.
For each $u\in V(G)\setminus\{v\}$, let $h(u)=[0,1]\times f(u)$ if $uv\in E(G)$ and 
$h(u)=[1/2,3/2]\times f(u)$ if $uv\not\in E(G)$.  Let $h(v)=[-1,0]\times [-M,M] \times \cdots \times [-M,M]$,
where $M$ is chosen large enough so that $f(u)\subseteq [-M,M] \times \cdots \times [-M,M]$ for every $u\in V(G)\setminus\{v\}$.
Then $h$ is a touching representation of $G$ by comparable boxes in $\mathbb{R}^{d+1}$.
\end{proof}

\subsection{Strong product}

Let $G\boxtimes H$ denote the \emph{strong product} of the graphs $G$
and $H$, i.e., the graph with vertex set $V(G)\times V(H)$ and with
distinct vertices $(u_1,v_1)$ and $(u_2,v_2)$ adjacent if and only if
$u_1$ is equal to or adjacent to $u_2$ in $G$
and $v_1$ is equal to or adjacent to $v_2$ in $H$.
To obtain a touching representation of $G\boxtimes
H$ it suffices to take a product of representations of $G$ and $H$, but
the resulting representation may contain incomparable boxes. 
Indeed, $\cbdim(G\boxtimes H)$ in general is not bounded by a function
of $\cbdim(G)$ and $\cbdim(H)$; for example, every star has comparable box dimension
at most two, but the strong product of the star $K_{1,n}$ with itself contains
$K_{n,n}$ as an induced subgraph, and thus its comparable box dimension is at least $\Omega(\log n)$.
However, as shown in the following lemma, this issue does not arise if the representation of $H$ consists of translates
of a single box; by scaling, we can without loss of generality assume this box is a unit hypercube.

\begin{lemma}\label{lemma-sp}
  Consider a graph $H$ having a touching representation $h$ in
  $\mathbb{R}^{d_H}$ by axis-aligned hypercubes of unit size.  Then for any graph
  $G$, the strong product $G\boxtimes H$ of these graphs has comparable box dimension at most
  $\cbdim(G) + d_H$.
\end{lemma}
\begin{proof}
  The proof simply consists in taking a product of the two
  representations.  Indeed, consider a touching respresentation $g$ of $G$ by 
  comparable boxes in $\mathbb{R}^{d_G}$, with
  $d_G=\cbdim(G)$, and the representation $h$ of $H$.  Let us define a
  representation $f$ of $G\boxtimes H$ in $\mathbb{R}^{d_G+d_H}$ as
  follows.
  \[f((u,v))[i]=\begin{cases}
  g(u)[i]&\text{ if $i\le d_G$}\\
  h(v)[i-d_G]&\text{ if $i > d_G$}
  \end{cases}\]
  Consider distinct vertices $(u,v)$ and $(u',v')$ of $G\boxtimes H$.
  The boxes $g(u)$ and $g(u')$ are comparable, say $g(u)\sqsubseteq g(u')$.  Since $h(v')$
  is a translation of $h(v)$, this implies that $f((u,v))\sqsubseteq f((u',v'))$. Hence, the boxes
  of the representation $f$ are pairwise comparable.

  The boxes of the representations $g$ and $h$ have pairwise disjoint interiors.
  Hence, if $u\neq u'$, then there exists $i\le d_G$ such that the interiors
  of the intervals $f((u,v))[i]=g(u)[i]$ and $f((u',v'))[i]=g(u')[i]$ are disjoint;
  and if $v\neq v'$, then there exists $i\le d_H$ such that the interiors
  of the intervals $f((u,v))[i+d_G]=h(v)[i]$ and $f((u',v'))[i+d_G]=h(v')[i]$ are disjoint.
  Consequently, the interiors of boxes $f((u,v))$ and $f((u',v'))$ are pairwise disjoint.
  Moreover, if $u\neq u'$ and $uu'\not\in E(G)$, or if $v\neq v'$ and $vv'\not\in E(G)$,
  then the intervals discussed above (not just their interiors) are disjoint for some $i$;
  hence, if $(u,v)$ and $(u',v')$ are not adjacent in $G\boxtimes H$, then $f((u,v))\cap f((u',v'))=\emptyset$.
  Therefore, $f$ is a touching representation of a subgraph of $G\boxtimes H$.

  Finally, suppose that $(u,v)$ and $(u',v')$ are adjacent in $G\boxtimes H$.
  Then there exists a point $p_G$ in the intersection of $g(u)$ and $g(u')$,
  since $u=u'$ or $uu'\in E(G)$ and $g$ is a touching representation of $G$;
  and similarly, there exists a point $p_H$ in the intersection of $h(v)$ and $h(v')$.
  Then $p_G\times p_H$ is a point in the intersection of $f((u,v))$ and $f((u',v'))$.
  Hence, $f$ is indeed a touching representation of $G\boxtimes H$.
\end{proof}

\subsection{Taking a subgraph}

The comparable box dimension of a subgraph of a graph $G$ may be larger than $\cbdim(G)$, see the end of this
section for an example. However, we show that the
comparable box dimension of a subgraph is at most exponential in the
comparable box dimension of the whole graph.  This is essentially
Corollary~25 in~\cite{subconvex}, but since the setting is somewhat
different and the construction of~\cite{subconvex} uses rotated boxes,
we provide details of the argument.

\begin{lemma}\label{lemma-subg}
If $G$ is a subgraph of a graph $G'$, then $\cbdim(G)\le \cbdim(G')+\frac12 \chi^2_s(G')$.
\end{lemma}
\begin{proof}
As we can remove the boxes that represent the vertices, we can assume $V(G')=V(G)$.
Let $f$ be a touching representation of $G'$ by comparable boxes in $\mathbb{R}^d$, where $d=\cbdim(G')$.  Let $\varphi$
be a star coloring of $G'$ using colors $\{1,\ldots,c\}$, where $c=\chi_s(G')$.

For any distinct colors $i,j\in\{1,\ldots,c\}$, let $A_{i,j}\subseteq V(G)$ be the set of vertices $u$ of color $i$
such that there exists a vertex $v$ of color $j$ such that $uv\in E(G')\setminus E(G)$.  For each $u\in A_{i,j}$,
let $a_j(u)$ denote such a vertex $v$ chosen arbitrarily.

Let us define a representation $h$ by boxes in $\mathbb{R}^{d+\binom{c}{2}}$ by starting from the representation $f$ and,
for each pair $i<j$ of colors, adding a dimension $d_{i,j}$ and setting
\[h(v)[d_{i,j}]=\begin{cases}
[1/3,4/3]&\text{if $v\in A_{i,j}$}\\
[-4/3,-1/3]&\text{if $v\in A_{j,i}$}\\
[-1/2,1/2]&\text{otherwise.}
\end{cases}\]
Note that the boxes in this extended representation are comparable,
as in the added dimensions, all the boxes have size $1$.

Suppose $uv\in E(G)$, where $\varphi(u)=i$ and $\varphi(v)=j$ and say $i<j$.
We cannot have $u\in A_{i,j}$ and $v\in A_{j,i}$, as then $a_j(u)uva_i(v)$ would be a 4-vertex path in $G'$ in colors $i$ and $j$.
Hence, in any added dimension $d'$, we have $h(u)[d']=[-1/2,1/2]$ or $h(v)[d']=[-1/2,1/2]$,
and thus $h(u)[d']\cap h(v)[d']\neq\emptyset$.  
Since the boxes $f(u)$ and $f(v)$ touch, it follows that the boxes $h(u)$ and $h(v)$ touch as well.

Suppose now that $uv\not\in E(G)$.  If $uv\not\in E(G')$, then $f(u)$ is disjoint from $f(v)$, and thus $h(u)$ is disjoint from
$h(v)$.  Hence, we can assume $uv\in E(G')\setminus E(G)$, $\varphi(u)=i$, $\varphi(v)=j$ and $i<j$.  Then $u\in A_{i,j}$, $v\in A_{j,i}$,
$h(u)[d_{i,j}]=[1/3,4/3]$, $h(v)[d_{j,i}]=[-4/3,-1/3]$, and $h(u)\cap h(v)=\emptyset$.

Consequently, $h$ is a touching representation of $G$ by comparable boxes in dimension $d+\binom{c}{2}\le d+c^2 /2$.
\end{proof}

Let us now combine Lemmas~\ref{lemma-chrom} and \ref{lemma-subg}.

\begin{corollary}\label{cor-subg}
If $G$ is a subgraph of a graph $G'$, then $\cbdim(G)\le \cbdim(G')+2\cdot 81^{\cbdim(G')}\le 3\cdot 81^{\cbdim(G')}$.
\end{corollary}

Let us remark that an exponential increase in the dimension is unavoidable: We have $\cbdim(K_{2^d})=d$,
but the graph obtained from $K_{2^d}$ by deleting a perfect matching has comparable box dimension $2^{d-1}$. Indeed, for every pair $u,v$ of non-adjacent vertices there is a specific dimension $i$ such that their boxes span intervals $[a,b]$ and $[c,d]$ with $b<c$, while for every other box in the representation their $i^\text{th}$ interval contains $[b,c]$.

\subsection{Clique-sums}

A \emph{clique-sum} of two graphs $G_1$ and $G_2$ is obtained from
their disjoint union by identifying vertices of a clique in $G_1$ and
a clique of the same size in $G_2$ and possibly deleting some of the
edges of the resulting clique.  A \emph{full clique-sum} is a
clique-sum in which we keep all the edges of the resulting clique.
The main issue to overcome in obtaining a representation for a (full)
clique-sum is that the representations of $G_1$ and $G_2$ can be
``degenerate''. Consider e.g.\ the case that $G_1$ is represented by
unit squares arranged in a grid; in this case, there is no space to
attach $G_2$ at the cliques formed by four squares intersecting in a
single corner.  This can be avoided by increasing the dimension, but
we need to be careful so that the dimension stays bounded even after
an arbitrary number of clique-sums. We thus introduce the notion of
\emph{clique-sum extendable} representations.

\begin{definition}
Consider a graph $G$ with a distinguished clique $C^\star$, called the
\emph{root clique} of $G$. A touching representation $h$ of $G$
by (not necessarily comparable) boxes in $\mathbb{R}^d$ is called
\emph{$C^\star$-clique-sum extendable} if the following conditions hold for every sufficiently small $\varepsilon>0$.
\begin{itemize}
\item[] {\bf(vertices)} For each $u\in V(C^\star)$, there exists a dimension $d_u$,
  such that:
  \subitem\emph{(v0)} $d_u\neq d_{u'}$ for distinct $u,u'\in V(C^\star)$,
  \subitem\emph{(v1)} each vertex $u\in V(C^\star)$ satisfies $h(u)[d_u] = [-1,0]$ and
    $h(u)[i] = [0,1]$ for any dimension $i\neq d_u$, and
  \subitem\emph{(v2)} each vertex $v\notin V(C^\star)$ satisfies $h(v) \subset [0,1)^d$.
\item[] {\bf(cliques)} For every clique $C$ of $G$, there exists
  a point $p(C)\in [0,1)^d\cap \left( \bigcap_{v\in V(C)} h(v) \right)$
  such that, defining the \emph{clique box} $h^\varepsilon(C)$
  by setting $h^\varepsilon(C)[i] = [p(C)[i],p(C)[i]+\varepsilon]$ for every dimension
  $i$, the following conditions are satisfied:
%  \begin{itemize}
  \subitem\emph{(c1)} For any two cliques $C_1\neq C_2$, $h^\varepsilon(C_1) \cap
    h^\varepsilon(C_2) = \emptyset$ (equivalently, $p(C_1) \neq p(C_2)$).
  \subitem\emph{(c2)} A box $h(v)$ intersects $h^\varepsilon(C)$ if and only if
    $v\in V(C)$, and in that case their intersection is a facet of
    $h^\varepsilon(C)$ incident to $p(C)$.  That is, there exists a dimension $i_{C,v}$
    such that for each dimension $j$,
    \[h(v)[j]\cap h^\varepsilon(C)[j] = \begin{cases}
    \{p(C)[i_{C,v}] \}&\text{if $j=i_{C,v}$}\\
    [p(C)[j],p(C)[j]+\varepsilon]&\text{otherwise.}
    \end{cases}\]
%  \end{itemize}
\end{itemize}
\end{definition}
Note that the root clique can be empty, that is the
empty subgraph with no vertices.  In that case the clique is denoted
$\emptyset$.  Let $\ecbdim(G)$ be the minimum dimension such that $G$
has an $\emptyset$-clique-sum extendable touching representation by
comparable boxes.

Let us remark that a clique-sum extendable representation in dimension $d$ implies
such a representation in higher dimensions as well.
\begin{lemma}\label{lemma-add}
Let $G$ be a graph with a root clique $C^\star$ and let $h$ be
a $C^\star$-clique-sum extendable touching representation of $G$ by comparable boxes in $\mathbb{R}^d$.
Then $G$ has such a representation in $\mathbb{R}^{d'}$ for every $d'\ge d$.
\end{lemma}
\begin{proof}
It clearly suffices to consider the case that $d'=d+1$.
Note that the \textbf{(vertices)} conditions imply that $h(v')\sqsubseteq h(v)$ for every $v'\in V(G)\setminus V(C^\star)$
and $v\in V(C^\star)$.  We extend the representation $h$
by setting $h(v)[d+1] = [0,1]$ for $v\in V(C^\star)$ and $h(v)[d+1] = [0,\frac12]$ for $v\in V(G)\setminus V(C^\star)$.
The clique point $p(C)$ of each clique $C$ is extended by setting $p(C)[d+1] = \frac14$.
It is easy to verify that the resulting representation is $C^\star$-clique-sum extendable.
\end{proof}

The following lemma ensures that clique-sum extendable representations
behave well with respect to full clique-sums.
\begin{lemma}\label{lem-cs}
  Consider two graphs $G_1$ and $G_2$, given with a $C^\star_1$- and a
  $C^\star_2$-clique-sum extendable representations $h_1$ and $h_2$ by comparable boxes
  in $\mathbb{R}^{d_1}$ and $\mathbb{R}^{d_2}$,
  respectively. Let $G$ be the graph obtained by performing a full
  clique-sum of these two graphs on any clique $C_1$ of $G_1$, and on
  the root clique $C^\star_2$ of $G_2$.  Then $G$ admits a $C^\star_1$-clique
  sum extendable representation $h$ by comparable boxes in
  $\mathbb{R}^{\max(d_1,d_2)}$.
\end{lemma}

\begin{proof}
  By Lemma~\ref{lemma-add}, we can assume that $d_1=d_2$; let $d=d_1$.
  The idea is to translate (allowing also exchanges of dimensions) and
  scale $h_2$ to fit in $h_1^\varepsilon(C_1)$. Consider an $\varepsilon >0$
  sufficiently small so that $h_1^\varepsilon(C_1)$ satisfies all the
  \textbf{(cliques)} conditions, and such that $h_1^\varepsilon(C_1) \sqsubseteq
  h_1(v)$ for any vertex $v\in V(G_1)$.  Let $V(C_1)=\{v_1,\ldots,v_k\}$;
  without loss of generality, we can assume $i_{C_1,v_i}=i$ for $i\in\{1,\ldots,k\}$,
  and thus
  \[h_1(v_i)[j] \cap h_1^\varepsilon(C_1)[j] = \begin{cases}
  \{p_1(C_1)[i]\}&\text{ if $j=i$}\\
  [p_1(C_1)[j],p_1(C_1)[j]+\varepsilon]&\text{ otherwise.}
  \end{cases}\]

  Now let us consider $G_2$ and its representation $h_2$. Here the
  vertices of $C^\star_2$ are also denoted $v_1,\ldots,v_k$, and
  without loss of generality, the \textbf{(vertices)} conditions are
  satisfied by setting $d_{v_i}=i$ for $i\in\{1,\ldots,k\}$

  We are now ready to define $h$.  For $v\in V(G_1)$, we set $h(v)=h_1(v)$.
  We now scale and translate $h_2$ to fit inside $h_1^\varepsilon(C_1)$.
  That is, we fix $\varepsilon>0$ small enough so that
  \begin{itemize}
  \item the conditions \textbf{(cliques)} hold for $h_1$,
  \item $h_1^\varepsilon(C_1)\subset [0,1)^d$, and
  \item $h_1^\varepsilon(C_1)\sqsubseteq h_1(u)$ for every $u\in V(G_1)$,
  \end{itemize}
  and for each $v\in V(G_2) \setminus V(C^\star_2)$,
  we set $h(v)[i]=p_1(C_1)[i] + \varepsilon h_2(v)[i]$ for $i\in\{1,\ldots,d\}$.
  Note that the condition (v2) for $h_2$ implies $h(v)\subset h_1^\varepsilon(C_1)$.
  Each clique $C$ of $H$ is a clique of $G_1$ or $G_2$.
  If $C$ is a clique of $G_2$, we set $p(C)=p_1(C_1)+\varepsilon p_2(C)$,
  otherwise we set $p(C)=p_1(C)$. In particular, for subcliques of $C_1=C^\star_2$,
  we use the former choice.

  Let us now check that $h$ is a $C^\star_1$-clique sum extendable
  representation by comparable boxes. The fact that the boxes are
  comparable follows from the fact that those of $h_1$ and $h_2$
  are comparable and from the scaling of $h_2$:  By construction both
  $h_1(v) \sqsubseteq h_1(u)$ and $h_2(v) \sqsubseteq h_2(u)$ imply
  $h(v) \sqsubseteq h(u)$, and for any vertex $u\in V(G_1)$ and any
  vertex $v\in V(G_2) \setminus V(C^\star_2)$, we have $h(v) \subset h_1^\varepsilon(C_1) \sqsubseteq h(u)$.

  We now check that $h$ is a contact representation of $G$. For $u,v
  \in V(G_1)$ (resp. $u,v \in V(G_2) \setminus V(C^\star_2)$) it
  is clear that $h(u)$ and $h(v)$ have disjoint interiors, and that they
  intersect if and only if $h_1(u)$ and $h_1(v)$ intersect (resp. if
  $h_2(u)$ and $h_2(v)$ intersect). Consider now a vertex $u \in
  V(G_1)$ and a vertex $v \in V(G_2) \setminus V(C^\star_2)$. As
  $h(v)\subset h^\varepsilon(C_1)$, the condition (v2) for $h_1$ implies
  that $h(u)$ and $h(v)$ have disjoint interiors.
  
  Furthermore, if $uv\in E(G)$, then $u\in V(C_1)=V(C^\star_2)$, say $u=v_1$.
  Since $uv\in E(G_2)$, the intervals $h_2(u)[1]$ and $h_2(v)[1]$ intersect,
  and by (v1) and (v2) for $h_2$, we conclude that $h_2(v)[1]=[0,\alpha]$ for some positive $\alpha<1$.
  Therefore, $p_1(C_1)[1]\in h(v)[1]$.  Since $p_1(C_1)\in \bigcap_{x\in V(C_1)} h_1(x)$,
  we have $p_1(C_1)\in h(u)$, and thus $p_1(C_1)[1]\in h(u)[1]\cap h(v)[1]$.
  For $i\in \{2,\ldots,d\}$, note that $i\neq 1=i_{C_1,u}$, and thus
  by (c2) for $h_1$, we have $h_1^\varepsilon(C_1)[i]\subseteq h_1(u)[i]=h(u)[i]$.
  Since $h(v)[i]\subseteq h_1^\varepsilon(C_1)[i]$, it follows that $h(u)$ intersects $h(v)$.

  Finally, let us consider the $C^\star_1$-clique-sum extendability. The \textbf{(vertices)}
  conditions hold, since (v0) and (v1) are inherited from $h_1$, and
  (v2) is inherited from $h_1$ for $v\in V(G_1)\setminus V(C^\star_1)$
  and follows from the fact that $h(v)\subseteq h_1^\varepsilon(C_1)\subset [0,1)^d$
  for $v\in V(G_2)\setminus V(C^\star_2)$.  For the \textbf{(cliques)} condition (c1),
  the mapping $p$ inherits injectivity when restricted to cliques of $G_2$,
  or to cliques of $G_1$ not contained in $C_1$.  For any clique $C$ of $G_2$,
  the point $p(C)$ is contained in $h_1^\varepsilon(C_1)$, since $p_2(C)\in [0,1)^d$.
  On the other hand, if $C'$ is a clique of $G_1$ not contained in $C_1$, then there
  exists $v\in V(C')\setminus V(C_1)$, we have $p(C')=p_1(C')\in h_1(v)$, and
  $h_1(v)\cap h_1^\varepsilon(C_1)=\emptyset$ by (c2) for $h_1$.
  Therefore, the mapping $p$ is injective, and thus for sufficiently small $\varepsilon'>0$,
  we have $h^{\varepsilon'}(C)\cap h^{\varepsilon'}(C')=\emptyset$ for any distinct
  cliques $C$ and $C'$ of $G$.

  The condition (c2) of $h$ is (for sufficiently small $\varepsilon'>0$)
  inherited from the property (c2) of $h_1$ and $h_2$
  when $C$ is a clique of $G_2$ and $v\in V(G_2)\setminus V(C^\star_2)$, or
  when $C$ is a clique of $G_1$ not contained in $C_1$ and $v\in V(G_1)$.
  If $C$ is a clique of $G_1$ not contained in $C_1$ and $v\in V(G_2)\setminus V(C^\star_2)$,
  then by (c1) for $h_1$ we have $h_1^\varepsilon(C)\cap h_1^\varepsilon(C_1)=\emptyset$,
  and since $h^{\varepsilon'}(C)\subseteq h_1^\varepsilon(C)$ and $h(v)\subseteq h_1^\varepsilon(C_1)$,
  we conclude that $h(v)\cap h^{\varepsilon'}(C)=\emptyset$.
  It remains to consider the case that $C$ is a clique of $G_2$ and $v\in V(G_1)$.
  Note that $h^{\varepsilon'}(C)\subseteq h_1^\varepsilon(C_1)$.
  \begin{itemize}
  \item If $v\not\in V(C_1)$, then by the property (c2) of $h_1$, the box $h(v)=h_1(v)$ is disjoint from $h_1^\varepsilon(C_1)$,
  and thus $h(v)\cap h^{\varepsilon'}(C)=\emptyset$.
  \item Otherwise $v\in V(C_1)=V(C^\star_2)$, say $v=v_1$.
  Note that by (v1), we have $h_2(v)=[-1,0]\times [0,1]^{d-1}$.
  \begin{itemize}
  \item If $v\not\in V(C)$, then by the property (c2) of $h_2$, the box $h_2(v)$ is disjoint from $h_2^\varepsilon(C)$.
  Since $h_2^\varepsilon(C)[i]\subseteq[0,1]=h_2(v)[i]$ for $i\in\{2,\ldots,d\}$,
  it follows that $h_2^\varepsilon(C)[1]\subseteq (0,1)$, and thus $h^{\varepsilon'}(C)[1]\subseteq h_1^\varepsilon(C_1)[1]\setminus\{p(C_1)[1]\}$.
  By (c2) for $h_1$, we have $h(v)[1]\cap h_1^\varepsilon(C_1)[1]=h_1(v)[1]\cap h_1^\varepsilon(C_1)[1]=p(C_1)[1]$,
  and thus $h(v)\cap h^{\varepsilon'}(C)=\emptyset$.
  \item If $v\in V(C)$, then by the property (c2) of $h_2$, the intersection of
  $h_2(v)[1]=[-1,0]$ and $h_2^\varepsilon(C)[1]\subseteq [0,1)$ is the single point $p_2(C)[1]=0$,
  and thus $p(C)[1]=p_1(C_1)[1]$ and $h^{\varepsilon'}(C)[1]=[p_1(C_1)[1],p_1(C_1)[1]+\varepsilon']$.
  Recall that the property (c2) of $h_1$ implies $h(v)[1]\cap h_1^\varepsilon(C_1)[1]=\{p(C_1)[1]\}$,
  and thus $h(v)[1]\cap h^{\varepsilon'}(C)[1]=\{p(C)[1]\}$.  For $i\in\{2,\ldots, d\}$,
  the property (c2) of $h_1$ implies $h_1^\varepsilon(C_1)[i]\subseteq h_1(v)[i]=h(v)[i]$, and
  since $h^{\varepsilon'}(C)[i]\subseteq h_1^\varepsilon(C_1)[i]$, it follows that
  $h^{\varepsilon'}(C)[i]\subseteq h(v)[i]$.
  \end{itemize}
  \end{itemize}
\end{proof}

%The proof is in the appendix, but the idea is to translate (allowing also exchanges of dimensions) and scale $h_2$ to fit in $h_1^\varepsilon(C_1)$.
The following lemma enables us to pick the root clique at the expense of increasing
the dimension by $\omega(G)$.
\begin{lemma}\label{lem-apex-cs}
  For any graph $G$ and any clique $C^\star$, the graph $G$ admits a
  $C^\star$-clique-sum extendable touching representation by comparable
  boxes in $\mathbb{R}^d$, for $d = |V(C^\star)| + \ecbdim(G\setminus V(C^\star))$.
\end{lemma}

\begin{proof}
  The proof is essentially the same as the one of
  Lemma~\ref{lemma-apex}.  Consider a $\emptyset$-clique-sum
  extendable touching representation $h'$ of $G\setminus V(C^\star)$ by
  comparable boxes in $\mathbb{R}^{d'}$, with $d' = \cbdim(G\setminus
  V(C^\star))$, and let $V(C^\star) = \{v_1,\ldots,v_k\}$. We now construct
  the desired representation $h$ of $G$ as follows. For each vertex
  $v_i\in V(C^\star)$, let $h(v_i)$ be the box in $\mathbb{R}^d$ uniquely determined
  by the condition (v1) with $d_{v_i} = i$. For each vertex $u\in V(G)\setminus V(C^\star)$,
  if $i\le k$ then let $h(u)[i] = [0,1/2]$ if $uv_i \in E(G)$, and $h(u)[i] =
  [1/4,3/4]$ if $uv_i \notin E(G)$. For $i>k$ we have $h(u)[i] =
  \alpha h'(u)[i-k]$, for some $\alpha>0$. The value $\alpha>0$
  is chosen sufficiently small so that $h(u)[i] \subset [0,1)$ whenever $u\notin V(C^\star)$. 
  We proceed similarly for the clique points. For any
  clique $C$ of $G$, if $i\le k$ then let $p(C)[i] = 0$ if $v_i \in V(C)$,
  and $p(C)[i] = 1/4$ if $v_i \notin V(C)$. For $i>k$ we refer to the clique point $p'(C')$ of $C'=C\setminus
  \{v_1,\ldots,v_k\}$, and we set $p(C)[i] = \alpha p'(C')[i-k]$. 
  
  By the construction, it is clear that $h$ is a touching representation of $G$.
  As $h'(u) \sqsubset h'(v)$ implies that $h(u) \sqsubset h(v)$, and as 
  $h(u) \sqsubset h(v_i)$ for every $u\in V(G)\setminus V(C^\star)$ and every 
  $v_i \in V(C^\star)$, we have that $h$ is a representation by comparable boxes.

  For the $C^\star$-clique-sum extendability, the \textbf{(vertices)} conditions hold by the construction.
  For the \textbf{(cliques)} condition (c1), let us consider distinct cliques $C_1$ and $C_2$
  of $G$ such that $|V(C_1)| \ge |V(C_2)|$, and let $C'_i=C_i\setminus V(C^\star)$. If $C'_1 = C'_2$,
  there is a vertex $v_i \in V(C_1) \setminus V(C_2)$, and $p(C_1)[i] = 0 \neq 1/4 = p(C_2)[i]$.
  Otherwise, if $C'_1 \neq C'_2$, then $p'(C'_1) \neq p'(C'_2)$, which implies
  $p(C_1) \neq p(C_2)$ by construction.

  For the \textbf{(cliques)} condition (c2), let us first consider a vertex $v\in V(G)\setminus V(C^\star)$ and
  a clique $C$ of $G$ containing $v$.  In the dimensions $i\in\{1,\ldots,k\}$, we always have
  $h^\varepsilon(C)[i] \subseteq h(v)[i]$. Indeed, if $v_i \in V(C)$, then
  $h^\varepsilon(C)[i] \subseteq [0,1/2] = h(v)[i]$, as in this case $v$ and $v_i$ are adjacent;
  and if $v_i \notin V(C)$, then $h^\varepsilon(C)[i] \subseteq [1/4,1/2] \subseteq h(v)[i]$.
  By the property (c2) of $h'$,
  we have $h^\varepsilon(C)[i] \subseteq h(v)[i]$ for every $i>k$, except one, 
  for which $h^\varepsilon(C)[i] \cap h(v)[i] = \{p(C)[i]\}$. 
  
  Next, let us consider a vertex $v\in V(G)\setminus V(C^\star)$ and a clique $C$ of $G$ not containing $v$. 
  As $v\notin V(C')$, the condition (c2) for $h'$ implies that $p'(C')$ is disjoint from $h'(v)$,
  and thus $p(C)$ is disjoint from $h(v)$.
  
  Finally, we consider a vertex $v_i \in V(C^\star)$.  Note that for any clique $C$ containing $v_i$,
  we have that $h^\varepsilon(C)[i] \cap h(v_i)[i] = [0,\varepsilon]\cap [-1,0] = \{0\}$, and $h^\varepsilon(C)[j] \subseteq [0,1] = h(v_i)[j]$
  for any $j\neq i$. For a clique $C$ that does not contain $v_i$ we have that 
  $h^\varepsilon(C)[i] \cap h(v_i)[i] \subset (0,1)\cap [-1,0] = \emptyset$. 
  Condition (c2) is thus fulfilled and this completes the proof of the lemma. 
\end{proof}

%  The proof is also in the appendix, but it essentially the same as the one of Lemma~\ref{lemma-apex}.
The following lemma provides an upper bound on $\ecbdim(G)$ in terms
of $\cbdim(G)$ and $\chi(G)$.
\begin{lemma}\label{lem-ecbdim-cbdim}
  For any graph $G$, $\ecbdim(G) \le \cbdim(G) + \chi(G)$.
\end{lemma}
\begin{proof}
  Let $h$ be a touching representation of $G$ by comparable boxes in
  $\mathbb{R}^d$, with $d=\cbdim(G)$, and let $c$ be a
  $\chi(G)$-coloring of $G$. We start with a slightly modified version
  of $h$. We first scale $h$ to fit in $(0,1)^d$, and for a
  sufficiently small real $\alpha>0$ we increase each box in $h$ by
  $2\alpha$ in every dimension, that is we replace $h(v)[i] = [a,b]$
  by $[a-\alpha,b+\alpha]$ for each vertex $v$ and dimension
  $i$. We choose $\alpha$ sufficiently small so that the boxes representing
  non-adjacent vertices remain disjoint, and thus the resulting representation $h_1$ is
  an intersection representation of the same graph $G$.  Moreover, observe that
  for every clique $C$ of $G$, the intersection $I_C=\bigcap_{v\in V(C)} h_1(v)$ is
  a box with non-zero edge lengths. For any clique $C$ of $G$, let
  $p_1(C)$ be a point in the interior of $I_C$ different from the points
  chosen for all other cliques.

  Now we add $\chi(G)$ dimensions to make the representation touching
  again, and to ensure some space for the clique boxes
  $h^\varepsilon(C)$. Formally we define $h_2$ as follows.
  \[h_2(u)[i]=\begin{cases}
  h_1(u)[i]&\text{ if $i\le d$}\\
  [1/5,3/5]&\text{ if $i>d$ and $c(u) < i-d$}\\
  [0,2/5]&\text{ if $i>d$ and $c(u) = i-d$}\\
  [2/5,4/5]&\text{ otherwise (if $c(u) > i-d > 0$)}
  \end{cases}\]
  For any clique $C$ of $G$, let $c(C)$ denote the color set $\{c(u)\ |\ u\in V(C)\}$.
  We now set
  \[p_2(C)[i]=\begin{cases}
  p_1(C)[i] &\text{ if $i\le d$}\\
  2/5 &\text{ if $i>d$ and $i-d \in c(C)$}\\
  1/2 &\text{ otherwise}
  \end{cases}
  \]
  As $h_2$ is an extension of $h_1$, and as in each dimension $j>d$,
  $h_2(v)[j]$ is an interval of length $2/5$ containing the point $2/5$ for every vertex $v$,
  we have that $h_2$ is an intersection representation of $G$ by comparable boxes.
  To prove that it is touching consider two adjacent 
  vertices $u$ and $v$ such that $c(u)<c(v)$, and let us note that $h_2(u)[d+c(u)] = [0,2/5]$
  and $h_2(v)[d+c(u)] = [2/5,4/5]$. 
  
  For the $\emptyset$-clique-sum extendability, the \textbf{(vertices)} conditions are void.
  For the \textbf{(cliques)} conditions, since $p_1$ is chosen to be injective, the mapping $p_2$
  is injective as well, implying that (c1) holds.
 
  Consider now a clique $C$ in $G$ and a vertex $v\in V(G)$.  If $c(v)\not\in c(C)$, then
  $h_2(v)[c(v)+d]=[0,2/5]$ and $p_2(C)[c(v)+d]=1/2$, implying that $h_2^{\varepsilon}(C)\cap h_2(v)=\emptyset$.
  If $c(v)\in c(C)$ but $v\not\in V(C)$, then letting $v'\in V(C)$ be the vertex of color $c(v)$,
  we have $vv'\not\in E(G)$, and thus $h_1(v)$ is disjoint from $h_1(v')$.  Since $p_1(C)$ is contained
  in the interior of $h_1(v')$, it follows that $h_2^{\varepsilon}(C)\cap h_2(v)=\emptyset$.
  Finally, suppose that $v\in C$.  Since $p_1(C)$ is contained in the interior of $h_1(v)$,
  we have $h_2^{\varepsilon}(C)[i] \subset h_2(v)[i]$ for every $i\le d$.  For $i>d$ distinct from $d+c(v)$,
  we have $p_2^{\varepsilon}(C)[i]\in\{2/5,1/2\}$ and $[2/5,3/5]\subseteq h_2(v)[i]$, and thus
  $h_2^{\varepsilon}(C)[i] \subset h_2(v)[i]$.  For $i=d+c(v)$, we have $p_2^{\varepsilon}(C)[i]=2/5$
  and $h_2(v)[i]=[0,2/5]$, and thus $h_2^{\varepsilon}(C)[i] \cap h_2(v)[i]=\{p_2^{\varepsilon}(C)[i]\}$.
  Therefore, (c2) holds.
\end{proof}

A touching representation of axis-aligned boxes in $\mathbb{R}^d$ is said \emph{fully touching} if any two intersecting boxes intersect on a $(d-1)$-dimensional box. Note that the construction above is fully touching.
Indeed, two intersecting boxes corresponding to vertices $u,v$ of colors $c(u) < c(v)$, only touch at coordinate $2/5$ in the $(d+c(u))^\text{th}$ dimension, while they fully intersect in every other dimension. This observation with Lemma~\ref{lemma-chrom} lead to the following.
\begin{corollary}
\label{cor-fully-touching}
Any graph $G$ has a fully touching representation of comparable axis-aligned boxes in $\mathbb{R}^d$, where $d= \cbdim(G) + 3^{\cbdim(G)}$.
\end{corollary}

Together, the lemmas from this section show that comparable box dimension is almost preserved by
full clique-sums.

\begin{corollary}
\label{cor-csum}
Let $\GG$ be a class of graphs of chromatic number at most $k$.  If $\GG'$ is the class
of graphs obtained from $\GG$ by repeatedly performing full clique-sums, then
\[\cbdim(\GG')\le \cbdim(\GG) + 2k.\]
\end{corollary}
\begin{proof}
Suppose a graph $G$ is obtained from $G_1, \ldots, G_m\in\GG$ by performing full clique-sums.
Without loss of generality, the labelling of the graphs is chosen so that we first
perform the full clique-sum on $G_1$ and $G_2$, then on the resulting graph and $G_3$, and so on.
Let $C^\star_1=\emptyset$ and for $i=2,\ldots,m$, let $C^\star_i$ be the root clique of $G_i$ on which it is
glued in the full clique-sum operation.  By Lemmas~\ref{lem-ecbdim-cbdim} and \ref{lem-apex-cs},
$G_i$ has a $C_i^\star$-clique-sum extendable touching representation by comparable boxes in $\mathbb{R}^d$,
where $d=\cbdim(\GG) + 2k$.  Repeatedly applying Lemma~\ref{lem-cs}, we conclude that
$\cbdim(G)\le d$.
\end{proof}

By Lemmas~\ref{lemma-chrom} and \ref{lemma-subg}, this gives the following bounds.
\begin{corollary}\label{cor-csump}
Let $\GG$ be a class of graphs of comparable box dimension at most $d$.
\begin{itemize}
\item The class $\GG'$ of graphs obtained from $\GG$ by repeatedly performing full clique-sums
has comparable box dimension at most $d + 2\cdot 3^d$.
\item The closure of $\GG'$ by taking subgraphs
has comparable box dimension at most $1250^d$.
\end{itemize}
\end{corollary}
\begin{proof}
The former bound directly follows from Corollary~\ref{cor-csum} and the bound on the chromatic number
from Lemma~\ref{lemma-chrom}.  For the latter one, we need to bound the star chromatic number of $\GG'$.
Suppose a graph $G$ is obtained from $G_1, \ldots, G_m\in\GG$ by performing full clique-sums.
For $i=1,\ldots, m$, suppose $G_i$ has an acyclic coloring $\varphi_i$ by at most $k$ colors.
Note that the vertices of any clique get pairwise different colors, and thus by permuting the colors,
we can ensure that when we perform the full clique-sum, the vertices that are identified have the same
color.  Hence, we can define a coloring $\varphi$ of $G$ such that for each $i$, the restriction of
$\varphi$ to $V(G_i)$ is equal to $\varphi_i$.  Let $C$ be the union of any two color classes of $\varphi$.
Then for each $i$, $G_i[C\cap V(G_i)]$ is a forest, and since $G[C]$ is obtained from these graphs
by full clique-sums, $G[C]$ is also a forest.  Hence, $\varphi$ is an acyclic coloring of $G$
by at most $k$ colors.  By~\cite{albertson2004coloring}, $G$ has a star coloring by at most $2k^2-k$ colors.
Hence, Lemma~\ref{lemma-chrom} implies that $\GG'$ has star chromatic number at most $2\cdot 25^d - 5^d$.
The bound on the comparable box dimension of subgraphs of graphs from $\GG'$ then follows from Lemma~\ref{lemma-subg}.
\end{proof}

\section{The strong product structure and minor-closed classes}

A \emph{$k$-tree} is any graph obtained by repeated full clique-sums on cliques of size $k$ from cliques of size at most $k+1$.
A \emph{$k$-tree-grid} is a strong product of a $k$-tree and a path.
An \emph{extended $k$-tree-grid} is a graph obtained from a $k$-tree-grid by adding at most $k$ apex vertices.
Dujmovi{\'c} et al.~\cite{DJM+} proved the following result.
\begin{theorem}\label{thm-prod}
Any graph $G$ is a subgraph of the strong product of a $k$-tree-grid and $K_m$, where
\begin{itemize}
\item $k=3$ and $m=3$ if $G$ is planar, and
\item $k=4$ and $m=\max(2g,3)$ if $G$ has Euler genus at most $g$.
\end{itemize}
Moreover, for every $t$, there exists an integer $k$ such that any
$K_t$-minor-free graph $G$ is a subgraph of a graph obtained by repeated clique-sums
from extended $k$-tree-grids.
\end{theorem}

Let us first bound the comparable box dimension of a graph in terms of
its Euler genus.  As paths and $m$-cliques admit touching
representations with hypercubes of unit size in $\mathbb{R}^{1}$ and
in $\mathbb{R}^{\lceil \log_2 m \rceil}$ respectively, by
Lemma~\ref{lemma-sp} it suffices to bound the comparable box
dimension of $k$-trees.

\begin{theorem}\label{thm-ktree}
  For any $k$-tree $G$,  $\cbdim(G) \le \ecbdim(G) \le k+1$.
\end{theorem}
\begin{proof}
  Let $H$ be a complete graph with $k+1$ vertices and let $C^\star$ be
  a clique of size $k$ in $H$.  By Lemma~\ref{lem-cs}, it suffices
  to show that $H$ has a $C^\star$-clique-sum extendable touching representation
  by hypercubes in $\mathbb{R}^{k+1}$.  Let $V(C^\star)=\{v_1,\ldots,v_k\}$.
  We construct the representation $h$ so that (v1) holds with $d_{v_i}=i$ for each $i$;
  this uniquely determines the hypercubes $h(v_1)$, \ldots, $h(v_k)$.
  For the vertex $v_{k+1} \in V(H)\setminus V(C^\star)$, we set $h(v_{k+1})=[0,1/2]^{k+1}$.
  This ensures that the \textbf{(vertices)} conditions holds.

  For the \textbf{(cliques)} conditions, let us set
the point $p(C)$ for every clique $C$ as follows:
\begin{itemize}
  \item $p(C)[i] = 0 $ for every $i\le k$ such that $v_i\in C$
  \item $p(C)[i] = \frac14 $ for every $i\le k$ such that $v_i\notin C$
  \item $p(C)[k+1] = \frac12 $ if $v_{k+1}\in C$
  \item $p(C)[k+1] = \frac34 $ if $v_{k+1}\notin C$
\end{itemize}
By construction, it is clear that for each vertex $v\in V(H)$, $p(C) \in h(v)$ if and only if
$v\in V(C)$.

For any two distinct cliques $C_1$ and $C_2$, the points $p(C_1)$ and
$p(C_2)$ are distinct.  Indeed, by symmetry we can assume that for some $i$
we have $v_i\in V(C_1)\setminus V(C_2)$, and this implies that $p(C_1)[i] < p(C_2)[i]$.
Hence, the condition (c1) holds.

Consider now a vertex $v_i$ and a clique $C$.  As we observed before, if $v_i\not\in V(C)$,
then $p(C) \not\in h(v_i)$, and thus $h^\varepsilon(C)$ and $h(v_i)$ are disjoint (for sufficiently small $\varepsilon>0$).
If $v_i\in C$, then the definitions ensure that $p(C)[i]$ is equal to the maximum of $h(v_i)[i]$,
and that for $j\neq i$, $p(C)[j]$ is in the interior of $h(v_i)[j]$, implying
$h(v_i)[j] \cap h^\varepsilon(C)[j] = [p(C)[j],p(C)[j]+\varepsilon]$ for sufficiently small $\varepsilon>0$.
\end{proof}
The \emph{treewidth} $\tw(G)$ of a graph $G$ is the minimum $k$ such that $G$ is a subgraph of a $k$-tree.
Note that actually the bound on the comparable box dimension of Theorem~\ref{thm-ktree}
extends to graphs of treewidth at most $k$.
\begin{corollary}\label{cor-tw}
Every graph $G$ satisfies $\cbdim(G)\le\tw(G)+1$.
\end{corollary}

\begin{proof}
Let $k=\tw(G)$.  Observe that there exists a $k$-tree $T$ with the root clique $C^\star$ such that $G\subseteq T-V(C^\star)$.
Inspection of the proof of Theorem~\ref{thm-ktree} (and Lemma~\ref{lem-cs}) shows that we obtain
a representation $h$ of $T-V(C^\star)$ in $\mathbb{R}^{k+1}$ such that
\begin{itemize}
\item the vertices are represented by hypercubes of pairwise different sizes,
\item if $uv\in E(T-V(C^\star))$ and $h(u)\sqsubseteq h(v)$, then $h(u)\cap h(v)$ is a facet of $h(u)$ incident
with its point with minimum coordinates, and
%\item for each vertex $u$ and each facet of $h(u)$ incident with its point with minimum coordinates, 
%there exists at most one vertex $v$ such that $uv\in E(T-V(C^\star))$ and $h(u)\sqsubseteq h(v)$.
\end{itemize}
If for some $u,v\in V(G)$, we have $uv\in E(T)\setminus E(G)$, where without loss of generality $h(u)\sqsubseteq h(v)$,
we now alter the representation by shrinking $h(u)$ slightly away from $h(v)$ (so that all other touchings are preserved).
Since the hypercubes of $h$ have pairwise different sizes, the resulting touching representation of $G$ is by comparable boxes.
\end{proof}

As every planar graph $G$ has a touching representation by cubes in
$\mathbb{R}^3$~\cite{felsner2011contact}, we have that $\cbdim(G)\le 3$.
For the graphs with higher Euler genus we can also derive upper
bounds.  Indeed, combining the previous observation on the
representations of paths and $K_m$, with Theorem~\ref{thm-ktree},
Lemma~\ref{lemma-sp}, and Corollary~\ref{cor-subg} we obtain:

\begin{corollary}\label{cor-genus}
For every graph $G$ of Euler genus $g$, there exists a supergraph $G'$
of $G$ such that $\cbdim(G')\le 6+\lceil \log_2 \max(2g,3)\rceil$.
Consequently, \[\cbdim(G)\le 3\cdot 81^7 \cdot \max(2g,3)^{\log_2 81}.\]
\end{corollary}

Similarly, we can deal with proper minor-closed classes.
\begin{proof}[Proof of Theorem~\ref{thm-minor}]
Let $\GG$ be a proper minor-closed class.  Since $\GG$ is proper, there exists $t$ such that $K_t\not\in \GG$.
By Theorem~\ref{thm-prod}, there exists $k$ such that every graph in $\GG$ is a subgraph of a graph obtained by repeated clique-sums
from extended $k$-tree-grids.  As we have seen, $k$-tree-grids have comparable box dimension at most $k+2$,
and by Lemma~\ref{lemma-apex}, extended $k$-tree-grids have comparable box dimension at most $2k+2$.
By Corollary~\ref{cor-csump}, it follows that $\cbdim(\GG)\le 1250^{2k+2}$.
\end{proof}

Note that the graph obtained from $K_{2n}$ by deleting a perfect matching has Euler genus $\Theta(n^2)$
and comparable box dimension $n$. It follows that the dependence of the comparable box dimension on the Euler genus cannot be
subpolynomial (though the degree $\log_2 81$ of the polynomial established in Corollary~\ref{cor-genus}
certainly can be improved).  The dependence of the comparable box dimension on the size of the forbidden minor that we
established is not explicit, as Theorem~\ref{thm-prod} is based on the structure theorem of Robertson and Seymour~\cite{robertson2003graph}.
It would be interesting to prove Theorem~\ref{thm-minor} without using the structure theorem.

\section{Fractional treewidth-fragility}

Suppose $G$ is a connected planar graph and $v$ is a vertex of $G$.  For an integer $k\ge 2$,
give each vertex at distance $d$ from $v$ the color $d\bmod k$.  Then deleting the vertices of any of the $k$ colors
results in a graph of treewidth at most $3k$. This fact (which follows from the result of Robertson and Seymour~\cite{rs3}
on treewidth of planar graphs of bounded radius) is (in the modern terms) the basis of Baker's technique~\cite{baker1994approximation}
for design of approximation algorithms.  However, even quite simple graph classes (e.g., strong products of three paths~\cite{gridtw})
do not admit such a coloring (where the removal of any color class results in a graph of bounded treewidth).
However, a fractional version of this coloring concept is still very useful in the design of approximation algorithms~\cite{distapx}
and applies to much more general graph classes, including all graph classes with strongly sublinear separators and bounded maximum degree~\cite{twd}.

We say that a class of graphs $\GG$ is \emph{fractionally treewidth-fragile} if there exists a function $f$ such that
for every graph $G\in\GG$ and integer $k\ge 2$, there exist sets $X_1, \ldots, X_m\subseteq V(G)$ such that
each vertex belongs to at most $m/k$ of them and $\tw(G-X_i)\le f(k)$ for every $i$
(equivalently, there exists a probability distribution on the set $\{X\subseteq V(G):\tw(G-X)\le f(k)\}$
such that $\text{Pr}[v\in X]\le 1/k$ for each $v\in V(G)$).
For example, the class of planar graphs is (fractionally) treewidth-fragile, since we can let $X_i$ consist of the
vertices of color $i-1$ in the coloring described at the beginning of the section.

Before going further, let us recall some notions about treewidth.
  A \emph{tree decomposition} of a graph $G$ is a pair
  $(T,\beta)$, where $T$ is a rooted tree and $\beta:V(T)\to 2^{V(G)}$
  assigns a \emph{bag} to each of its nodes, such that
\begin{itemize}
\item for each $uv\in E(G)$, there exists $x\in V(T)$ such that
  $u,v\in\beta(x)$, and
\item for each $v\in V(G)$, the set $\{x\in V(T):v\in\beta(x)\}$ is
  non-empty and induces a connected subtree of $T$.
\end{itemize}
For nodes $x,y\in V(T)$, we write $x\preceq y$ if $x=y$ or $x$ is a descendant of $y$ in $T$.
The \emph{width} of the tree decomposition is the maximum of the sizes of the bags minus $1$.  The \emph{treewidth} of a graph is the minimum
of the widths of its tree decompositions.  Let us remark that the value of treewidth obtained via this definition coincides
with the one via $k$-trees which we used in the previous section.

Our main result is that all graph classes of bounded comparable box dimension are fractionally treewidth-fragile.
We will show the result in a more general setting, motivated by concepts from~\cite{subconvex} and by applications to related
representations. The argument is motivated by the idea used in the approximation algorithms for disk graphs
by Erlebach et al.~\cite{erlebach2005polynomial}. Before introducing this more general setting, and as a warm-up, let us outline
how to prove that disk graphs of thickness $t$ are fractionally treewidth-fragile. Consider first unit disk graphs. 
By partitionning the plane with a random grid $\HH$, having squared cells of side-length $2k$, any unit disk has probability $1/2k$ 
to intersect a vertical (resp. horizontal) line of the grid. By union bound, any disk has probability at most $1/k$ to intersect 
the grid. Considering this probability distribution, let us now show that removing the disks intersected by the grid leads to a 
unit disk graph of bounded treewidth. Indeed, in such a graph any connected component corresponds to unit disks contained in the 
same cell of the grid. Such cell having area bounded by $4k^2$, there are at most $16tk^2/\pi$ disks contained in a cell. 
The size of the connected components being bounded, so is the treewidth. Note that this distribution also works if we are given
disks whose diameter lie in a certain range. If any diameter $\delta$ is such that $1/c \le \delta \le 1$, then the same process
with a random grid of $2k\times 2k$ cells, ensures that any disk is deleted with probability at most $1/k$, while now the
connected components have size at most $4tc^2k^2/\pi$. Dealing with arbitrary disk graphs (with any diameter $\delta$ being in the range 
$0< \delta \le 1$) requires to delete more disks. This is why each $(2k\times 2k)$-cell is now partitionned in a quadtree-like manner. 
Now a disk with diameter between $\ell /2$ and $\ell$ (with $\ell =1/2^i$ for some integer $i\ge 0$) is deleted if it is not contained 
in a $(2k\ell \times 2k\ell)$-cell of a quadtree. It is not hard to see that a disk is deleted with probability at most $1/k$.
To prove that the remaining graph has bounded treewidth one should consider the following tree decomposition $(T,\beta)$. The 
tree $T$ is obtained by linking the roots of the quadtrees we used (as trees) to a new common root. 
Then for a $(2k\ell \times 2k\ell)$-cell $C$, $\beta(C)$ contains all the disks of diameter at least $\ell/2$ intersecting $C$.
To see that such bag is bounded consider the $((2k+1)\ell \times (2k+1)\ell)$ square $C'$ centered on $C$, and note that any
disk in $\beta(C)$ intersects $C'$ on an area at least $\pi\ell^2/16$. This implies that $|\beta(C)| \le 16t(2k+1)^2 / \pi$.

Let us now give a detailed proof in a more general setting.
For a measurable set $A\subseteq \mathbb{R}^d$, let $\vol(A)$ denote the Lebesgue measure of $A$.
For two measurable subsets $A$ and $B$ of $\mathbb{R}^d$ and a positive integer $s$, we write $A\sqsubseteq_s B$
if for every $x\in B$, there exists a translation $A'$ of $A$ such that $x\in A'$ and $\vol(A'\cap B)\ge \tfrac{1}{s}\vol(A)$.
Note that for two boxes $A$ and $B$, we have $A\sqsubseteq_1 B$ if and only if $A\sqsubseteq B$.
An \emph{$s$-comparable envelope representation} $(\iota,\omega)$ of a graph $G$ in $\mathbb{R}^d$ consists of 
two functions $\iota,\omega:V(G)\to 2^{\mathbb{R}^d}$ such that for some ordering $v_1$, \ldots, $v_n$ of vertices of $G$,
\begin{itemize}
\item for each $i$, $\omega(v_i)$ is a box, $\iota(v_i)$ is a measurable set, and $\iota(v_i)\subseteq \omega(v_i)$,
\item if $i<j$, then $\omega(v_j)\sqsubseteq_s \iota(v_i)$, and
\item if $i<j$ and $v_iv_j\in E(G)$, then $\omega(v_j)\cap \iota(v_i)\neq\emptyset$.
\end{itemize}
We say that the representation has \emph{thickness at most $t$} if for every point $x\in \mathbb{R}^d$, there
exist at most $t$ vertices $v\in V(G)$ such that $x\in\iota(v)$.

For example:
\begin{itemize}
\item If $f$ is a touching representation of $G$ by comparable boxes in $\mathbb{R}^d$, then
$(f,f)$ is a $1$-comparable envelope representation of $G$ in $\mathbb{R}^d$ of thickness at most $2^d$.
\item If $f$ is a touching representation of $G$ by balls in $\mathbb{R}^d$ and letting $\omega(v)$ be
the smallest axis-aligned hypercube containing $f(v)$, then there exists a positive integer $s_d$ depending only on $d$ such that
$(f,\omega)$ is an $s_d$-comparable envelope representation of $G$ in $\mathbb{R}^d$ of thickness at most $2$.
\end{itemize}

\begin{theorem}\label{thm-twfrag}
For positive integers $t$, $s$, and $d$, the class of graphs
with an $s$-comparable envelope representation in $\mathbb{R}^d$ of thickness at most $t$
is fractionally treewidth-fragile, with a function $f(k) = O_{t,s,d}\bigl(k^{d}\bigr)$.
\end{theorem}
\begin{proof}
For a positive integer $k$, let $f(k)=(2ksd+2)^dst$.
Let $(\iota,\omega)$ be an $s$-comparable envelope representation of a graph $G$
in $\mathbb{R}^d$ of thickness at most $t$, and let $v_1$, \ldots, $v_n$ be the corresponding ordering of the vertices of $G$.
Let us define $\ell_{i,j}\in \mathbb{R}^+$ for $i=1,\ldots, n$ and $j\in\{1,\ldots,d\}$ as an approximation of $ksd|\omega(v_i)[j]|$ such that $\ell_{i-1,j} / \ell_{i,j}$ is a positive integer. Formally
it is defined as follows.
\begin{itemize}
\item Let $\ell_{1,j}=ksd|\omega(v_1)[j]|$.
\item For $i=2,\ldots, n$, let $\ell_{i,j} = \ell_{i-1,j}$, if
  $\ell_{i-1,j} < ksd|\omega(v_i)[j]|$, and otherwise let
  $\ell_{i,j}$ be lowest fraction of $\ell_{i-1,j}$ that is
  greater than $ksd|\omega(v_i)[j]|$, formally $\ell_{i,j} =
  \min\{\ell_{i-1,j}/b \ |\ b\in
  \mathbb{N}^+ \text{ and } \ell_{i-1,j}/b \ge ksd|\omega(v_i)[j]|\}$.
\end{itemize}
Let the real $x_j\in [0,\ell_{1,j}]$ be chosen uniformly at random,
and let $\HH^i_j$ be the set of hyperplanes in $\mathbb{R}^d$
consisting of the points whose $j$-th coordinate is equal to
$x_j+m\ell_{i,j}$ for some $m\in\mathbb{Z}$. As $\ell_{i,j}$ is a
multiple of $\ell_{i',j}$ whenever $i\le i'$, we have that $\HH^i_j
\subseteq \HH^{i'}_j$ whenever $i\le i'$.  For $i\in\{1,\ldots,n\}$,
the \emph{$i$-grid} is $\HH^i=\bigcup_{j=1}^d \HH^i_j$, and we let the
$0$-grid $\HH^0=\emptyset$.  Similarly as above we have that $\HH^i
\subseteq \HH^{i'}$ whenever $i\le i'$.

We let $X\subseteq V(G)$ consist of the vertices $v_a\in V(G)$ such
that the box $\omega(v_a)$ intersects some hyperplane $H\in \HH^a$,
that is such that $x_j+m\ell_{a,j}\in \omega(v_a)[j]$, for some
$j\in\{1,\ldots,d\}$ and some $m\in \mathbb{Z}$.  First, let us argue
that $\text{Pr}[v_a\in X]\le 1/k$.  Indeed, the set $[0,\ell_{1,j}]\cap \bigcup_{m\in\mathbb{Z}} (\omega(v_a)[j]-m\ell_{a,j})$
has measure $\tfrac{\ell_{1,j}}{\ell_{a,j}}\cdot |\omega(v_a)[j]|$, implying that for fixed $j$, this happens with probability
$|\omega(v_a)[j]|/\ell_{a,j}$.  Let $a'$ be the largest integer such
that $a'\le a$ and $\ell_{a',j} < \ell_{a'-1,j}$ if such an index exists,
and $a'=1$ otherwise; note that $\ell_{a,j}=\ell_{a',j}\ge ksd|\omega(v_{a'})[j]|$.  Moreover, since
$\omega(v_a)\sqsubseteq_s\iota(v_{a'})\subseteq \omega(v_{a'})$, we have $\omega(v_a)[j]\le s\omega(v_{a'})[j]$.
Combining these inequalities,
\[\frac{|\omega(v_a)[j]|}{\ell_{a,j}}\le \frac{s\omega(v_{a'})[j]}{ksd|\omega(v_{a'})[j]|}=\frac{1}{kd}.\]
By the union bound, we conclude that $\text{Pr}[v_a\in X]\le 1/k$.

Let us now bound the treewidth of $G-X$.  
For $a\ge 0$, an \emph{$a$-cell} is a maximal connected subset of $\mathbb{R}^d\setminus \bigl(\bigcup_{H\in \HH^a} H\bigr)$.
A set $C\subseteq\mathbb{R}^d$ is a \emph{cell} if it is an $a$-cell for some $a\ge 0$.
A cell $C$ is \emph{non-empty} if there exists $v\in V(G-X)$ such that $\iota(v)\subseteq C$.
Note that there exists a rooted tree $T$ whose vertices are
the non-empty cells and such that for $x,y\in V(T)$, we have $x\preceq y$ if and only if $x\subseteq y$.
For each non-empty cell $C$, let us define $\beta(C)$ as the set of vertices $v_i\in V(G-X)$ such that
$\iota(v)\cap C\neq\emptyset$ and $C$ is an $a$-cell for some $a\ge i$.

Let us show that $(T,\beta)$ is a tree decomposition of $G-X$.  For each $v_j\in V(G-X)$, the $j$-grid is disjoint from $\omega(v_j)$,
and thus $\iota(v_j)\subseteq \omega(v_j)\subset C$ for some $j$-cell $C\in V(T)$ and $v_j\in \beta(C)$.  Consider now an edge $v_iv_j\in E(G-X)$, where $i<j$.
We have $\omega(v_j)\cap \iota(v_i)\neq\emptyset$, and thus $\iota(v_i)\cap C\neq\emptyset$ and $v_i\in \beta(C)$.
Finally, suppose that $v_j\in C'$ for some $C'\in V(T)$.  Then $C'$ is an $a$-cell for some $a\ge j$, and since
$\iota(v_j)\cap C'\neq\emptyset$ and $\iota(v_j)\subset C$, we conclude that $C'\subseteq C$, and consequently $C'\preceq C$.
Moreover, any cell $C''$ such that $C'\preceq C''\preceq C$ (and thus $C'\subseteq C''\subseteq C$) is an $a'$-cell
for some $a'\ge j$ and $\iota(v_j)\cap C''\supseteq \iota(v_j)\cap C'\neq\emptyset$, implying $v_j\in\beta(C'')$.
It follows that $\{C':v_j\in\beta(C')\}$ induces a connected subtree of $T$.

Finally, let us bound the width of the decomposition $(T,\beta)$.  Let $C$ be a non-empty cell and let $a$ be maximum such that $C$
is an $a$-cell.  Then $C$ is an open box with sides of lengths $\ell_{a,1}$, \ldots, $\ell_{a,d}$.  Consider $j\in\{1,\ldots,d\}$:
\begin{itemize}
\item If $a=1$, then $\ell_{a,j}=ksd |\omega(v_a)[j]|$.
\item If $a>1$ and $\ell_{a,j}=\ell_{a-1,j}$, then $\ell_{a,j}=\ell_{a-1,j}<2ksd|\omega(v_a)[j]|$ (otherwise $\ell_{a,j}=\ell_{a-1,j}/b$ for some integer $b\ge 2$).
\item If $a>1$ and $\ell_{a,j} < \ell_{a-1,j}$, then $\ell_{a-1,j}\ge b\times ksd|\omega(v_a)[j]|$ for some integer $b\ge 2$. Now let $b$ be the greatest such integer (that is such that $\ell_{a-1,j} < (b+1)\times ksd|\omega(v_a)[j]|$) and note that
\[\ell_{a,j}=\frac{\ell_{a-1,j}}{b}<\tfrac{b+1}{b}ksd|\omega(v_a)[j]|<\tfrac{3}{2}ksd|\omega(v_a)[j]|.\]
\end{itemize}
Hence, $\ell_{a,j}<2ksd |\omega(v_a)[j]|$.  Let $C'$ be the box with the same center as $C$ and with $|C'[j]|=(2ksd+2)|\omega(v_a)[j]|$.
For any $v_i\in \beta(C)\setminus\{v_a\}$, we have $i\le a$ and $\iota(v_i)\cap C\neq\emptyset$, and since $\omega(v_a)\sqsubseteq_s \iota(v_i)$,
there exists a translation $B_i$ of $\omega(v_a)$ that intersects $C\cap \iota(v_i)$ and such that $\vol(B_i\cap\iota(v_i))\ge \tfrac{1}{s}\vol(\omega(v_a))$.
Note that as $B_i$ intersects $C$, we have that $B_i\subseteq C'$.
Since the representation has thickness at most $t$,
\begin{align*}
\vol(C')&\ge \vol\left(C'\cap \bigcup_{v_i\in \beta(C)\setminus\{v_a\}} \iota(v_i)\right)\\
&\ge \vol\left(\bigcup_{v_i\in \beta(C)\setminus\{v_a\}} B_i\cap\iota(v_i)\right)\\
&\ge \frac{1}{t}\sum_{v_i\in \beta(C)\setminus\{v_a\}} \vol(B_i\cap\iota(v_i))\\
&\ge \frac{\vol(\omega(v_a))(|\beta(C)|-1)}{st}.
\end{align*}
Since $\vol(C')=(2ksd+2)^d\vol(\omega(v_a))$, it follows that
\[|\beta(C)|-1\le (2ksd+2)^dst=f(k),\]
as required.
\end{proof}

The proof that (generalizations of) graphs with bounded comparable box dimensions have sublinear separators in~\cite{subconvex}
is indirect; it is established that these graphs have polynomial coloring numbers, which in turn implies they have polynomial
expansion, which then gives sublinear separators using the algorithm of Plotkin, Rao, and Smith~\cite{plotkin}.
The existence of sublinear separators is known to follow more directly from fractional treewidth-fragility.  Indeed, since $\text{Pr}[v\in X]\le 1/k$,
there exists $X\subseteq V(G)$ such that $\tw(G-X)\le f(k)$ and $|X|\le |V(G)|/k$. The graph $G-X$ has a balanced separator of size
at most $\tw(G-X)+1$, which combines with $X$ to a balanced separator of size at most $V(G)|/k+f(k)+1$ in $G$.
Optimizing the value of $k$ (choosing it so that $V(G)|/k=f(k)$), we obtain the following corollary of Theorem~\ref{thm-twfrag}.

\begin{corollary}
For positive integers $t$, $s$, and $d$, every graph $G$
with an $s$-comparable envelope representation in $\mathbb{R}^d$ of thickness at most $t$
has a sublinear separator of size $O_{t,s,d}\bigl(|V(G)|^{\tfrac{d}{d+1}}\bigr).$
\end{corollary}

\section{Acknowledgement}
This research was carried out at the workshop on Geometric Graphs and Hypergraphs organized by Yelena Yuditsky and Torsten Ueckerdt in September 2021.  We would like to thank the organizers and all participants for creating a friendly and productive environment.

\bibliography{main}

%\appendix

%\section{Omitted proofs}

%\newtheorem*{lemma-A}{Lemma~\ref{lem-cs}}
%\begin{lemma-A}
%  Consider two graphs $G_1$ and $G_2$, given with a $C^\star_1$- and a
%  $C^\star_2$-clique-sum extendable representations $h_1$ and $h_2$ by comparable boxes
%  in $\mathbb{R}^{d_1}$ and $\mathbb{R}^{d_2}$,
%  respectively. Let $G$ be the graph obtained by performing a full
%  clique-sum of these two graphs on any clique $C_1$ of $G_1$, and on
%  the root clique $C^\star_2$ of $G_2$.  Then $G$ admits a $C^\star_1$-clique
%  sum extendable representation $h$ by comparable boxes in
%  $\mathbb{R}^{\max(d_1,d_2)}$.
%\end{lemma-A}

%\newtheorem*{lemma-B}{Lemma~\ref{lem-apex-cs}}
%\begin{lemma-B}
%  For any graph $G$ and any clique $C^\star$, the graph $G$ admits a
%  $C^\star$-clique-sum extendable touching representation by
%  comparable boxes in $\mathbb{R}^d$, for $d = |V(C^\star)| +
%  \ecbdim(G\setminus V(C^\star))$.
%\end{lemma-B}

%\newtheorem*{corollary-C}{Corollary~\ref{cor-tw}}
%\begin{corollary-C}
%Every graph $G$ satisfies $\cbdim(G)\le\tw(G)+1$.
%\end{corollary-C}
\end{document}